\documentclass[letterpaper, 10 pt, conference]{ieeeconf}

\IEEEoverridecommandlockouts
\usepackage{graphicx}

\usepackage{amssymb}
\usepackage{subcaption}
\usepackage{amsmath} 
\usepackage{ntheorem}
\usepackage{xcolor}
\newtheorem{remark}{Remark}
\newtheorem{definition}{Definition}
\newtheorem{lemma}{Lemma}
\newtheorem{theorem}{Theorem}
\usepackage{mathtools}
\newtheorem{proposition}{Proposition}

\newtheorem{assumption}{Assumption}
\usepackage{microtype}
\newcommand{\sg}[1]{{\color{black} #1}}

\newcommand\MATRIX[1]{\begin{bmatrix} #1 \end{bmatrix}}

\title{
Price-Coordinated Mean Field Games with State Augmentation for Decentralized Battery Charging
}

\author{Nour Al Dandachly,  Shuang Gao and Roland Malhamé
\thanks{*This work is supported in part by  NSERC and FRQNT.}
\thanks{The authors are with the Department of Electrical Engineering, Polytechnique Montreal, and GERAD, Montreal, Quebec, Canada,  H3T 1J4.
         Email: {\tt\small $\{$nour.al-dandachly, shuang.gao,roland.malhame$\}$@polymtl.ca}}}

\begin{document}

\maketitle
\thispagestyle{empty}
\pagestyle{empty}

\begin{abstract}
This paper addresses the decentralized coordinated charging problem for a large
population of battery storage agents  (e.g. residential batteries, electrical vehicles, charging station batteries) using  Mean Field Game (MFG). 
Agents are assumed to have {affine dynamics} and are coupled through a price that is continuous and monotonically increasing with respect to the difference between the average charging power and the grid's desired average charging power.
An important
modeling feature of the proposed framework is the state augmentation, that is,  the
charging power is treated as a state variable and its rate of change (i.e. the ramp
rate)  as the control input. 
The resulting MFG equilibrium is characterized by two nonlinearly 
coupled forward-backward differential equations.
The existence and uniqueness of the MFG equilibrium is established for any continuous and monotonically increasing nonlinear price function 
without additional restrictions on the time horizon.
Moreover, in the special case where the price is  affine in the average charging power,  we further simplify the characterization of the MFG equilibrium strategy via two separate Riccati equations, both of which admit unique positive semi-definite solutions without additional assumptions.
\end{abstract}

\section{INTRODUCTION}

The large-scale deployment of battery storage assets ranging from residential batteries and electric vehicles (EVs) to charging station batteries is important in current energy decarbonization strategies \cite{IEA2024, Arbabzadeh2019, Golombek2022}. However, uncoordinated charging of a large population of such devices can impose significant stress on the power grid, increasing peak demand and inducing sharp aggregate load variations that complicate grid balancing and voltage regulation. Such concerns motivate the design of scalable and decentralized coordination mechanisms to shape collective charging behavior for demand response such as valley-filling \cite{MaCallawayHiskens2013, GanTopcuLow2013} and peak shaving.

Mean Field Game (MFG) theory, introduced by Lasry and Lions \cite{LasryLions2007} and independently by Huang, Caines, and Malham\'e \cite{HuangCainesMalhame2006}, provides a framework for designing decentralized strategies for large populations of strategically competitive agents. A particularly relevant class of models within this framework involves agent interactions mediated through price mechanisms. Price-based MFGs have been used for electricity market price formation
\cite{gomes2021mean}, applications involving renewable energy \cite{campi2026,shrivats2022mean}, storage coordination in smart grids 
\cite{Alasseur2020}, household electric heaters \cite{tchuendom2024}, and EV charging 
\cite{Hedel2024, Tajeddini2019,Couillet2012}. Related formulations for EV charging include a deterministic mean field control approach based on a broadcast incentive signal \cite{Grammatico2016}, and an aggregative game approach \cite{Paccagnan2016}. Other MFG formulations that do not rely on a mean-field-dependent price have been applied to demand management of thermostatically controlled loads 
\cite{BagagioloBauso2014} and EV charigng coordination \cite{BausoNamerikawa2019,Lin2022,Caballe2026}. Within the literature on EV charging, in the formulations \cite{Hedel2024,Tajeddini2019,Couillet2012,Grammatico2016,
Paccagnan2016,Lin2022,Caballe2026},
the charging power is treated as a control input, so its rate of change 
enters neither the state dynamics nor the cost directly.

In contrast, we introduce a state augmentation that treats the charging 
power as a state variable and its rate of change as the control input.

We formulate the coordinated battery-charging problem as a price-coordinated MFG problem with affine dynamics and quadratic costs. Following the state augmentation described above, interactions among agents are modeled through a price signal that couples to the population average charging power. We note that such a modeling choice in the context of centralized control with a mean-field-dependent price is explored in a recent work \cite{fabini2026trading}.

\textbf{Contribution}: First, we introduce the state augmentation  for the applications of MFGs to  decentralized price-coordinated  charging  of battery storage agents, where each agent's charging power is treated as a state variable. Such a state augmentation leads to  simplified representations of the MFG strategy. Second, we establish the existence and uniqueness of the MFG equilibrium, which holds for any continuous and monotonically increasing price function.
An important feature of this result is that, different from the linear-quadratic 
MFGs \cite[Theorem~3.3]{Bensoussan2016}, \cite{HuangCainesMalhame2007} and from 
prior battery charging MFG results relying on a linear price and a contraction 
condition \cite{Tajeddini2019}, the existence and uniqueness of the MFG equilibrium 
does not involve a contraction condition that restricts the (price) coupling 
strength and time horizon.
The proof relies on constructing an auxiliary convex optimal control problem whose necessary and sufficient conditions for optimality coincide with those for the existence and uniqueness of the MFG equilibrium. Third, for the special case where the price function is affine, we establish a simplified representation of the MFG equilibrium by identifying two  separate
 Riccati equations, both of which admit unique positive semi-definite solutions.
All results are established at the level of the limit MFG, corresponding to the infinite-population regime.

\textbf{Notation}: $\mathbb{R}$ denotes the set of real numbers. $\mathbb{R}^n$ denotes the $n$-dimensional Euclidean space. 
For a matrix $A$,
$A \geq 0$ (resp.\ $A > 0$)  means that $A$ is positive 
semidefinite (resp.\ positive definite). $A^\top$ denotes the transpose of $A$. 
$\mathbb{E}[\cdot]$ denotes the expectation. For a scalar function $V(\sg{z},t)$ with 
$\sg{z} \in \mathbb{R}^n$, $\nabla_{\sg{z}} V$ and $\nabla^2_\sg{z} V$ denote its gradient and Hessian 
with respect to the variable $z$, respectively. $(\Omega,\mathcal{F},
\mathbb{P})$ 
denotes a complete probability space.
For a Hilbert space $\mathbb{H}$, we use $C([0, T];\mathbb{H})$ 
(resp.\ $C_1([0, T];\mathbb{H})$) to denote the space of continuous (resp.\ continuously 
differentiable) functions from $[0,T]$ to $\mathbb{H}$, and $L^2([0,T];\mathbb{R})$ to 
denote the space of square-integrable functions from $[0,T]$ to $\mathbb{R}$.

\section{Model and Problem Formulation}
\label{sec:problem_formulation}

We consider a population of $N$ energy storage agents (e.g. electric vehicles, residential
batteries, or charging station batteries) for a finite time horizon $[0, T]$. The dynamics of agent $i\in\{1,\dots,N\}$ can be modeled by 
\begin{equation}\label{eq:component_dynamics}
\begin{aligned}
dx_{1,i}(t) &= \kappa\, x_{2,i}(t)\, dt - b(t)\,dt + \sigma_1 dw_{1,i}(t),\\
dx_{2,i}(t) &= u_i(t)\, dt + \sigma_2\, dw_{2,i}(t),
\end{aligned}
\end{equation}
with initial conditions $x_{1,i}(0)=x_{1,i,0}$ and
$x_{2,i}(0)=x_{2,i,0}$.
Here $x_{1,i}(t)$ is the state of charge (SOC) of agent $i$ in kWh and $x_{2,i}(t)$ denotes its charging power in kW at time $t$. The processes $\{w_{1,i}(\cdot)\}_{i=1}^N$ and $\{w_{2,i}(\cdot)\}_{i=1}^N$ are independent standard Brownian motions, and   $u_i(t)\in\mathbb{R}$ is the control input representing the rate of change of the charging power at time $t$,  referred to as the ramp rate \cite{Sukumar2018}.
The parameter $\kappa\in (0,1)$ is the conversion efficiency from the charging power to the increase of SOC per unit time. The noise intensity $\sigma_1>0$ for the SOC dynamics represents stochastic variations in energy consumption due to auxiliary loads. The noise intensity $\sigma_2>0$ for the charging power dynamics represents fluctuations in the charger output (e.g. due to control loop imperfections, voltage variations, and transient dynamics of the power electronic equipment). The function $b(\cdot):[0,T]\to\mathbb{R}_{+}$ is assumed to be a deterministic piece-wise continuous function representing the EV-baseline power consumption due to auxiliary loads and, more generally, any exogenous power drain on the battery.
Let 
$$
x_i(t)=[x_{1,i}(t)~ x_{2,i}(t)]^\top \text{ and } w_i(t)=[w_{1,i}(t)~  w_{2,i}(t)]^\top.  
$$ Then the dynamics of agent $i$ in \eqref{eq:component_dynamics} can be equivalently represented  by 
\begin{align}
dx_i(t) = \big(A x_i(t) + B u_i(t) + f(t)\big)\,dt+ \Sigma\,dw_i(t),\label{eq:dynamics}
\end{align}
where
\[
\begin{aligned}
    & A=\MATRIX{0&\kappa\\0&0},
	B=\MATRIX{0\\1}, 
	 f(t)=\MATRIX{-b(t)\\0},
	\Sigma=\MATRIX{\sigma_1 & 0 \\ 0& \sigma_2}
\end{aligned}
\]
with initial condition $x_i(0) = x_{i,0} \in \mathbb{R}^2$, where $x_{i,0} = [x_{1,i,0},\; x_{2,i,0}]^\top$. The initial states $\{x_{i,0}\}_{i=1}^N$ are assumed to be independent and identically distributed (i.i.d.) with known mean $\mathbb{E}[x_{i,0}]$ and finite variance.

\subsection{Cost Functional with Price Coupling}

Let $r_x(t) = [r_{x,1}(t), r_{x,2}(t)]^\top$ denote the desired reference at time $t$, where $r_{x,1}(t)$ and $r_{x,2}(t)$ are the target SOC and charging power, respectively.

For agent $i\in \{1,\dots, N\}$, the cost functional is given by
\begin{equation}
J_i^N (u_i)
	=
	\mathbb{E}\!\left[
	\int_0^T \ell_i^N (x_i(t), u_i(t), t)\, dt
	+ g_i(x_i(T))
	\right],
	\label{eq:cost}
\end{equation}
where $\ell_i^N:\mathbb{R}^2 \times \mathbb{R} \times [0,T]\to\mathbb{R}$ is the instantaneous cost defined as
\begin{equation}
	\begin{split}
	    \ell_i^N(x_i(t),u_i(t),t) =\;&
	    \frac{1}{2}\big(x_i(t)-r_x(t)\big)^\top
	    Q\big(x_i(t)-r_x(t)\big)\\
	    & + p^N(t)\, e_2^\top x_i(t)
	    + \frac{1}{2}u_i(t)^\top R u_i(t),
	\end{split}
	\label{eq:running_cost}
\end{equation}
and $g_i:\mathbb{R}^2\to\mathbb{R}$ is the terminal cost given by
\begin{equation}
	    g_i(x_i(T)) =
	    \frac{1}{2}\big(x_i(T)-r_x(T)\big)^\top
	    Q_T\big(x_i(T)-r_x(T)\big),
	    \label{eq:terminal_cost}
\end{equation}
with $Q\geq 0$,  $Q_T \geq 0$,  $R>0$. Let $e_2: =(0,1)^\top$ so that $e_2^\top x_i(t)=x_{2,i}(t)$, and $p^N(t)\in\mathbb{R}$ is the price signal at time $t$ given by
\begin{equation}
\label{eq:mean_field_price}
p^N(t)=\alpha\big(\bar x_2^N(t)-r_g(t)\big)
\end{equation}
where
$\bar{x}_2^{N}(t):=\frac{1}{N}\sum_{i=1}^N x_{2,i}(t)$
is the  population average charging power,
$r_{\mathrm{g}}(t)$ is the grid operator's target charging power per agent, 
and $\alpha(\cdot) : \mathbb{R} \to \mathbb{R}$ is a price function.

We introduce the following assumptions.
\begin{assumption} \label{ass:reference_signals}
    The reference signals $r_x$ and $r_g$ are continuous over $[0,T)$.
\end{assumption}

\begin{assumption} \label{ass:price_function}
    The price function $\alpha(\cdot) : \mathbb{R} \to \mathbb{R}$ is  monotonically 
    increasing and continuous.
\end{assumption}

The price signal $p^N(t)$ in \eqref{eq:mean_field_price} provides a 
coordination mechanism: under Assumption~\ref{ass:price_function}, it incentivizes agents to align their average charging power 
with the grid target $r_{\mathrm{g}}(t)$. When the population average exceeds the 
target, the price increases, discouraging further charging; when it falls below, 
the price decreases, encouraging charging.

The running cost in \eqref{eq:running_cost} is equivalent to
\begin{equation}
\begin{aligned}
\ell_i^N(x_i(t),u_i(t),t)
	=\;&
	\frac{1}{2}x_i(t)^\top Q x_i(t)
	- x_i(t)^\top Q r_x(t) \\
	&\quad
	+ p^N(t)\, e_2^\top x_i(t) \\
	&\quad
	+ \frac{1}{2}u_i(t)^\top R u_i(t)
	+ c(t),
\end{aligned}
\label{eq:running_cost_equiv}
\end{equation}
where $c(t)=\frac{1}{2}r_x(t)^\top Q r_x(t)$ is a deterministic scalar function of time, independent of the state and control variables.
The coupling appears only in the linear term $p^N(t) e_2^\top x_i(t)$, while the quadratic term $\frac{1}{2}x_i(t)^\top Q x_i(t)$ is independent of the population average.

  The charging power is modeled as a continuous variable in the dynamics \eqref{eq:dynamics},
  which is possible since modern chargers permit fine-grained adjustment of charging current \cite{Richardson2012}.

\section{Mean Field Game Equilibrium}
\label{sec:mfg}

We adopt the fixed-point approach for MFGs developed in \cite{HuangCainesMalhame2006} to analyze the large-population limit of the $N$-agent stochastic game introduced in Section~\ref{sec:problem_formulation}. While the individual agent's optimization problem has a linear-quadratic structure, the mean field coupling through the nonlinear price function $\alpha$ distinguishes our problem from classical linear-quadratic MFGs. As $N\to\infty$, under the standard MFG framework \cite{HuangCainesMalhame2006}, the empirical state average converges to a deterministic mean field trajectory $\bar{x}(\cdot)$ (assuming the limit exists).
For a generic agent $i$, the mean field trajectory satisfies
\[
\begin{aligned}
\bar{x}(t)
&= \mathbb{E}[x_i(t)]
= \lim_{N\to\infty} \bar{x}^{\,N}(t), \\
\bar{x}_2(t)
&= e_2^\top \bar{x}(t)
= \mathbb{E}[x_{2,i}(t)].
\end{aligned}
\]
where $x_i(t)$ denotes the state of agent $i$ at time $t$, $\bar{x}^N(t)= \frac{1}{N}\sum_{i=1}^N x_i(t) \in \mathbb{R}^2 $  is the empirical state average.  
Let $\mathcal{F}_t^i := \sigma\{x_i(s) : 0 \leq s \leq t\}$ denote 
the natural filtration generated by agent $i$'s state process, $\mathcal{U}_i$ denote the set of admissible control processes
\[
\mathcal{U}_i =
\left\{
\begin{array}{c}
u_i:[0,T]\times\Omega\to\mathbb{R}\ \big|\ u_i(\cdot)\ 
\text{is $\{\mathcal{F}_t^i\}$-adapted,} \\[4pt]
\mathbb{E}\!\int_0^T |u_i(t)|^2\,dt < \infty
\end{array}
\right\}.
\]

Given a deterministic mean charging power trajectory $\bar{x}_2(\cdot)$, the best-response problem for agent $i$ is to find a control $u_i \in \mathcal{U}_i$ that minimizes the cost functional
\begin{equation}
J_i(u_i;\bar{x}_2(\cdot))
	=
	\mathbb{E}\!\left[
	\int_0^T \ell_i\big(x_i(t),u_i(t),t\big)\,dt
	+
	g_i\big(x_i(T)\big)
	\right]
	\label{eq:limit-cost}
\end{equation}
subject to the state dynamics \eqref{eq:dynamics}. The running cost
$\ell_i(\cdot)$ is obtained from its finite-population counterpart $\ell_i^N(\cdot)$ 
in \eqref{eq:running_cost}
by replacing the empirical average charging power $\bar{x}_2^N(\cdot)$ with its deterministic limit $\bar{x}_2(\cdot)$. 
That is,
\begin{equation}\label{eq:running_cost_limit}
\begin{split}
\ell_i(x_i(t),u_i(t),t) =\;&
\tfrac{1}{2}\big(x_i(t)-r_x(t)\big)^\top Q\big(x_i(t)-r_x(t)\big)\\
&+ p(t)\,e_2^\top x_i(t) + \tfrac{1}{2}u_i(t)^\top Ru_i(t),
\end{split}
\end{equation}
with $p(t) = \alpha\big(\bar{x}_2(t) - r_g(t)\big)$. The interaction among agents enters the cost functional only through the mean charging power trajectory $\bar{x}_2(\cdot)$, which appears in the price signal $p(t)$.

\begin{remark}
The MFG problem above differs from classical LQ-MFG problems \cite[Theorem~3.3]{Bensoussan2016}, \cite{HuangCainesMalhame2007} where the coupling typically appears in the quadratic tracking term. Nevertheless, for a given mean field trajectory $\bar{x}_2(\cdot)$, the individual agent's best-response problem is a standard linear-quadratic optimal control problem with a time-varying  coefficient $p(t)$ even  when $\alpha(\cdot)$ is a general nonlinear function.
\end{remark}

An \emph{MFG equilibrium} is a pair $(\gamma^\star, \bar{x}^\star)$ consisting of a feedback strategy $\gamma^\star: [0,T] \times \mathbb{R}^2 \to \mathbb{R}$ and a deterministic mean field trajectory $\bar{x}^\star: [0,T] \to \mathbb{R}^2$ that satisfy the following conditions:
\begin{enumerate}
   \item[\textbf{(i)}] \textbf{Optimality:}
Given the mean charging power trajectory
$\bar{x}_2^\star(\cdot) = e_2^\top \bar{x}^\star(\cdot)$,
the control $u_i^\star(t) = \gamma^\star(t, x_i^\star(t))$ is optimal for
agent $i$'s best-response problem \eqref{eq:limit-cost}, i.e.,
\[
u_i^\star
\in
\arg\min_{u_i \in \mathcal{U}_i}
J_i\!\big(u_i;\,\bar{x}_2^\star(\cdot)\big).
\]

    \item[\textbf{(ii)}] \textbf{Consistency:} When agent $i$ uses the feedback law $\gamma^\star$, the induced state trajectory $x_i^\star(\cdot)$ (where $u_i^\star(t) = \gamma^\star(t, x_i^\star(t))$) satisfies the consistency condition
    \begin{equation}
    \label{eq:consistency_condition}
    \bar{x}_2^\star(t) = \mathbb{E}\!\left[ x_{2,i}^\star(t) \right], \qquad \forall\, t\in[0,T].
    \end{equation}
\end{enumerate}

\subsection{Best Response Problem}
\label{subsec:hjb}
To characterize the MFG equilibrium, we solve the best-response problem 
\eqref{eq:limit-cost} for agent $i$ given a fixed candidate mean charging power 
trajectory $\bar{x}_2(\cdot)$. 
Consider the value function
\begin{equation}
\label{eq:value_function}
V_i(t,z)
=
\inf_{u_i\in \mathcal{U}_i}
\mathbb{E}\!\Big[ 
\int_t^T \ell_i\big(x_i(s),u_i(s),s\big)\,ds
+
g_i\big(x_i(T)\big)
\Big],
\end{equation}
with $(t,z)\in[0,T]\times\mathbb{R}^2$, subject to the dynamics 
\eqref{eq:dynamics} with initial condition $x_i(t)=z$ at time $t$.

For the linear-quadratic optimal control problem in \eqref{eq:limit-cost}, the value function 
$V_i$ satisfies the Hamilton--Jacobi--Bellman (HJB) equation 
\cite[Ch.~4]{YongZhou1999}
\begin{equation}
\label{eq:hjb}
\begin{aligned}
-\partial_t V_i(t,z)
& =
\inf_{u_i\in\mathbb{R}}
\Big\{
\nabla_{z} V_i(t,z)^\top\!\big(Az+Bu_i+f(t)\big) \\
&\quad
+\tfrac12\,u_i^\top R u_i 
+\tfrac12\big(z-r_x(t)\big)^\top
Q\big(z-r_x(t)\big) \\
&\quad
+p(t)\,e_2^\top z
\Big\} +\tfrac12\,\mathrm{Tr}\!\big(\Sigma\Sigma^\top\nabla_{z}^2 V_i(t,z)\big),
\end{aligned}
\end{equation}
with terminal condition
\begin{equation}
\label{eq:hjb_terminal}
V_i(T,z)
=
\tfrac12\big(z-r_x(T)\big)^\top
Q_T\big(z-r_x(T)\big),
\end{equation}
where $p(t) = \alpha(\bar{x}_2(t) - r_g(t))$ is the price signal. 
Exploiting the linear-quadratic structure, we consider the quadratic ansatz
\begin{equation}
\label{eq:value_ansatz}
V_i(t,z)
=
\frac12\,z^\top P(t)z + s(t)^\top z + \phi(t),\quad z\in\mathbb{R}^2
\end{equation} where $P(t)\in\mathbb{R}^{2\times 2}$ is symmetric, $s(t)\in\mathbb{R}^2$, 
and $\phi(t)\in\mathbb{R}$. Since all agents are homogeneous, the value function $V_i \equiv V$ is 
identical for every agent $i$. Accordingly, the coefficients $P(t)$, 
$s(t)$, and $\phi(t)$ in the ansatz \eqref{eq:value_ansatz} are 
independent of $i$ and hence carry no agent index.
 The optimal control in \eqref{eq:hjb}
for agent $i$ is then given by 
\begin{equation}
\label{eq:opt_control}
u_i^\star(t) = \gamma^\star\!\big(t,\,x_i^\star(t)\big)
= -R^{-1}B^\top\!\big(P(t)\,x_i^\star(t)+s(t)\big),
\end{equation}
with the feedback law $\gamma^\star(t,z) = -R^{-1}B^\top(P(t)z+s(t))$.
Substituting \eqref{eq:value_ansatz} and 
\eqref{eq:opt_control} into \eqref{eq:hjb} and matching coefficients of 
equal order in $z$ yields the following equations
\begin{subequations}
\label{eq:riccati_system}
\begin{align}
-\dot P(t)
&=
A^\top P(t)+P(t)A
-P(t)BR^{-1}B^\top P(t)+Q,
\label{eq:riccati_P}\\[0.4em]
-\dot s(t)
&=
\Big(A^\top-P(t)BR^{-1}B^\top\Big)s(t)\notag\\
&\quad
+P(t)f(t)
+\;p(t)\,e_2
-\;Qr_x(t),
\label{eq:riccati_s}\\[0.4em]
-\dot\phi(t)
&=
-\frac12\, s(t)^\top BR^{-1}B^\top s(t)
+s(t)^\top f(t)\notag\\
&\quad
+\frac12\,\mathrm{Tr}\!\big(\Sigma\Sigma^\top P(t)\big)
+\frac12\,r_x(t)^\top Qr_x(t),
\label{eq:riccati_phi}
\end{align}
\end{subequations}
with the terminal conditions
\begin{equation}
\label{eq:terminal_conditions}
\begin{aligned}
P(T)   &= Q_T, \quad 
s(T)   = -\,Q_T r_x(T),\\
\phi(T)&= \tfrac12\,r_x(T)^\top Q_T r_x(T).
\end{aligned}
\end{equation}

\subsection{Closed-Loop Dynamics and Mean Field Consistency}
\label{subsec:consistency}
Define the closed-loop matrix
\begin{equation}
\label{eq:closed_loop_matrix}
A_{\mathrm{cl}}(t) = A-BR^{-1}B^\top P(t).
\end{equation}
Under the optimal control law \eqref{eq:opt_control} and using 
\eqref{eq:riccati_P}--\eqref{eq:riccati_s}, let $x_i^\star(\cdot
)$ denote the 
optimal state process of agent $i$. Then $x_i^\star(t)$ evolves as
\begin{equation}
\label{eq:closed_loop}
dx_i^\star(t)
=
\Big(A_{\mathrm{cl}}(t)x_i^\star(t)-BR^{-1}B^\top s(t)+f(t)\Big)\,dt
+\Sigma\,dw_i(t).
\end{equation}
Taking the expectation of the solution to  \eqref{eq:closed_loop} and denoting
$\bar{x}(t)=\mathbb{E}[x_i^\star(t)]$
yields the  deterministic mean dynamics
\begin{equation}
\label{eq:mean_ode}
\dot{\bar{x}}(t)
=
A_{\mathrm{cl}}(t)\bar{x}(t)
- BR^{-1}B^\top s(t)
+ f(t),
\quad
\bar{x}(0)=\mathbb{E}[x_{i,0}].
\end{equation}
Since $\bar{x}_2(t) = e_2^\top\bar{x}(t)$, the price signal can be expressed as
\begin{equation}
\label{eq:mfg_fixed_point}
p(t)=\alpha\big(e_2^\top\bar{x}(t)-r_g(t)\big),
\end{equation}
which, together with \eqref{eq:mean_ode} and \eqref{eq:riccati_s}, 
forms a coupled system that must be solved simultaneously.

\subsection{MFG Strategy and the Associated TPBVP}

From subsections~\ref{subsec:hjb} and \ref{subsec:consistency}, the MFG 
equilibrium strategy,
shared by all agents, is given by the feedback law 
\begin{equation}
    \gamma^\star(t,z) = -R^{-1}B^\top\big(P(t)z + s(t)\big),
    \label{eq:mfg_optimal_control}
\end{equation}
and each agent $i$ applies the control $u_i^\star(t) = \gamma^\star(t, x_i^\star(t))$ where $P(t)$ solves the standard Riccati equation
\begin{equation}\label{eq:mfgs_riccati_P}
\begin{aligned}
-\dot{P}(t)
&= A^\top P(t) + P(t)A - P(t)BR^{-1}B^\top P(t) + Q, \\[0.5ex]
P(T) &= Q_T,
\end{aligned}
\end{equation}
 $s(t)$ is the adjoint state satisfying the backward linear differential equation
\begin{equation}\label{eq:mfgs_adjoint}
\begin{aligned}
-\dot{s}(t)
&= \big(A^\top - P(t)BR^{-1}B^\top\big)s(t) + P(t)f(t) \\[0.5ex]
&\quad + \alpha\big(e_2^\top\bar{x}(t) - r_g(t)\big)e_2
- Qr_x(t),
\end{aligned}
\end{equation}
with terminal condition
\begin{equation}
    s(T) = -Q_T r_x(T),
    \label{eq:mfgs_adjoint_terminal}
\end{equation}
and the evolution of the population mean state is  given by
\begin{equation}
    \dot{\bar{x}}(t) = \big(A - BR^{-1}B^\top P(t)\big)\bar{x}(t) - BR^{-1}B^\top s(t) + f(t),
    \label{eq:mfgs_mean_forward}
\end{equation}
with the initial condition
\begin{equation}
    \bar{x}(0) = \mathbb{E}[x_{i,0}].
    \label{eq:mfgs_mean_initial}
\end{equation}

\subsection{Existence and Uniqueness of the TPBVP Solution}

The MFG strategy is characterized by the coupled system \eqref{eq:mfgs_adjoint}--\eqref{eq:mfgs_mean_initial}, which forms a two-point boundary value problem (TPBVP). In the following, we establish that this TPBVP admits a unique solution. The proof relies on analyzing an auxiliary deterministic convex optimal control problem, whose  solution is uniquely  characterized by the TPBVP.

Consider the auxiliary problem of finding a control $u(\cdot) \in L^2([0,T];\mathbb{R})$ that minimizes the cost functional
\begin{equation} \label{eq:aux_cost}
\begin{aligned}
J(u)&=\int_0^T\Big(\tfrac12 y(t)^\top Q y(t)+\tfrac12 u(t)^\top Ru(t)\\
& \quad  +\Phi(e_2^\top y(t)-r_g(t))-r_x(t)^\top Q y(t)\Big)dt\\
&\quad+\tfrac12 (y(T)-r_x(T))^\top Q_T(y(T)-r_x(T))
\end{aligned}
\end{equation}
where $\Phi(d) = \int_0^d \alpha(\tau)\,d\tau$, and $y(\cdot)$ is the state trajectory governed by the dynamics
\begin{equation}\label{eq:aux_dynamics}
\dot{y}(t) = Ay(t) + Bu(t) + f(t), \quad y(0)=\mathbb{E}[x_{i,0}].
\end{equation}

\begin{lemma}[Existence, Uniqueness, and Strict Convexity]
\label{lem:aux_existence_uniqueness}
Under Assumptions~1 and~2, the cost functional $J(u)$ defined 
in \eqref{eq:aux_cost} is strictly convex on $L^2([0,T];\mathbb{R})$, 
and there exists a unique optimal control $u^* \in L^2([0,T];\mathbb{R})$ 
minimizing $J(u)$ subject to the dynamics \eqref{eq:aux_dynamics}.
\end{lemma}

\begin{proof}
We first show that $J(u)$ is strictly convex in the control $u\in L^2([0,T];\mathbb{R})$. 
Existence and uniqueness of the minimizer then follow from continuity 
and coercivity of $J(u)$ on the reflexive Banach space $L^2([0,T];\mathbb{R})$.

Since the dynamics \eqref{eq:aux_dynamics} are affine in $(y,u)$, the 
control-to-state map $u(\cdot) \mapsto y(\cdot)$ is affine: for any 
$u_1(\cdot), u_2(\cdot) \in L^2([0,T];\mathbb{R})$ and $\lambda\in(0,1)$, 
the convex combination $u_\lambda(t) = (1-\lambda)u_1(t) + \lambda u_2(t)$ 
produces the state trajectory $y_\lambda(t) = (1-\lambda)y_1(t)+\lambda y_2(t)$ 
for all $t\in[0,T]$, where $y_i(\cdot)$ is the solution of 
\eqref{eq:aux_dynamics} associated with $u_i(\cdot)$, $i=1,2$.

Since $\alpha(\cdot)$ is continuous and nondecreasing 
following Assumption~\ref{ass:price_function}, its primitive 
$\Phi(d)=\int_0^d\alpha(\tau)\,d\tau$ satisfies $\Phi'=\alpha$, and 
since $\alpha$ is nondecreasing, $\Phi'$ is nondecreasing. Therefore, 
$\Phi$ is convex on $\mathbb{R}$. Hence, for each fixed $t\in[0,T]$, 
the mapping $y(t)\mapsto\Phi(e_2^\top y(t)-r_g(t))$ is convex in 
$y(t)\in\mathbb{R}^2$, as the composition of the convex function $\Phi$ 
with the affine map $y(t)\mapsto e_2^\top y(t)-r_g(t)$. Each term in the running cost integrand in \eqref{eq:aux_cost} is jointly 
convex in $(y(t),u(t))\in\mathbb{R}^2\times\mathbb{R}$ for each fixed 
$t\in[0,T]$. Indeed, 
$\frac{1}{2}y(t)^\top Qy(t)$ is convex in $y(t)$ since $Q\geq 0$; 
$\Phi(e_2^\top y(t)-r_g(t))$ is convex in $y(t)$; 
$-r_x(t)^\top Qy(t)$ is affine in $y(t)$; and 
$\frac{1}{2}u(t)^\top Ru(t)$ is strictly convex in $u(t)$ since $R>0$. 
Since a function convex in one component and independent of the other is 
convex on the product space $\mathbb{R}^2\times\mathbb{R}$, the running 
cost integrand is jointly convex in $(y(t),u(t))$.

Using the affine dependence $y_\lambda(t) = (1-\lambda)y_1(t)+\lambda y_2(t)$, 
the joint convexity of the running cost integrand, and integrating over $[0,T]$, 
we obtain $J(u_\lambda) \le (1-\lambda)J(u_1) + \lambda J(u_2).$
Moreover, the mapping $u \mapsto \int_0^T \tfrac{1}{2}u(t)^\top Ru(t)\,dt$ is strictly convex on $L^2([0,T];\mathbb{R})$ since $R>0$, whereas the remaining
terms in $J$ are convex after composition with the affine control-to-state map.
Therefore, for $u_1\neq u_2$,
\[
J(u_\lambda) < (1-\lambda)J(u_1) + \lambda J(u_2),
\]
and hence $J$ is strictly convex on $L^2([0,T];\mathbb{R})$.

Furthermore, one can establish the coercivity of $J(u)$; indeed, we show that $J(u) \to +\infty$ as $\|u(\cdot)\|_{L^2} \to \infty$.

The affine dynamics \eqref{eq:aux_dynamics} admit the explicit solution
\begin{equation*}
y(t) = e^{At}\mathbb{E}[x_{i,0}]
+ \int_0^t e^{A(t-\tau)}\bigl(Bu(\tau)+f(\tau)\bigr)\,d\tau,
\end{equation*}
for all $t \in [0,T]$.
Taking norms and applying the Cauchy--Schwarz inequality, there exist 
constants $C_1, C_2 > 0$, depending on $A$, $B$, $T$, and 
$\mathbb{E}[x_{i,0}]$, such that
\begin{equation}
\|y\|_C \le C_1\|u\|_{L^2} + C_2,
\label{eq:state_bound}
\end{equation}
where $\|y\|_C := \sup_{t \in [0,T]}\|y(t)\|_2$ denotes the uniform norm on $C([0,T];\mathbb{R}^2)$.
We now bound $J(u)$ from below. The terminal cost satisfies $\frac{1}{2}(y(T)-r_x(T))^\top Q_T (y(T)-r_x(T)) \ge 0,$ and $\int_0^T \frac{1}{2}y^\top(t) Qy(t)\,dt \ge 0$ since $Q_T, Q\ge 0$. The control cost satisfies $\int_0^T \frac{1}{2}u^\top(t) Ru(t)\,dt = \frac{R}{2}\|u\|_{L^2}^2$. For the nonlinear term, since $\Phi$ is convex, the first-order inequality gives $\Phi(d) \ge \Phi(0) + \alpha(0)d$ for all $d \in \mathbb{R}$, where $\alpha(0)$ may be positive or negative. Substituting $d(t) = e_2^\top y(t) - r_g(t)$ and using \eqref{eq:state_bound} 
 yields a lower bound linear in $\|u\|_{L^2}$. The linear term $-\int_0^T r_x^\top(t) Qy(t)\,dt$ is similarly bounded from below linearly in $\|u\|_{L^2}$ via Cauchy--Schwarz. Combining all terms, there exist constants $C, D > 0$ such that
\[
J(u) \ge \frac{R}{2}\|u\|_{L^2}^2 - C\|u\|_{L^2} - D.
\]
Since $R > 0$, the quadratic term dominates. This implies that when $\|u\|_{L^2} \to +\infty$,  $J(u)$ diverges to $+\infty$, and $J(u)$ is coercive.

Since $J(u)$ is strictly convex, continuous, and coercive on the reflexive Banach space $L^2([0,T];\mathbb{R})$, standard results in convex analysis guarantee the existence of a unique minimizer $u^* \in L^2([0,T];\mathbb{R})$ \cite[Cor.~3.23]{Brezis2011}, \cite[Ch.~I, Thm.~2.4]{fleming1975deterministic}.
\end{proof}
\begin{remark}
The convexity of the cost functional has been used to
establish the uniqueness of optimal control in \cite{Firoozi2020, 
shrivats2022mean} for linear-quadratic optimal control and in \cite[Ch.~II, Thm.~11.6]{fleming1975deterministic} for linear systems with convex costs. However, these results do not apply directly to our problems.
\end{remark}

Define the Hamiltonian $H: \mathbb{R}^2 \times \mathbb{R} \times 
\mathbb{R}^2 \to \mathbb{R}$ associated with the auxiliary problem 
\eqref{eq:aux_cost}--\eqref{eq:aux_dynamics} by
\begin{multline}
H(t,y,u,\lambda) = \tfrac{1}{2}y^\top Qy - r_x(t)^\top Qy
+ \Phi\!\bigl(e_2^\top y - r_g(t)\bigr) \\
+ \tfrac{1}{2}u^\top Ru + \lambda^\top\!\bigl(Ay + Bu + f(t)\bigr).
\label{eq:hamiltonian}
\end{multline}
By Pontryagin's Minimum Principle \sg{(PMP)}
\cite[Ch.~II, Thm.~5.1]{fleming1975deterministic}, for the unique minimizer $u^*$ of $J(u)$ guaranteed by Lemma 1, there exists a costate trajectory $\lambda^* \in C([0,T];\mathbb{R}^2)$ such that, denoting by $y^*$ the state trajectory of \eqref{eq:aux_dynamics} driven by $u^*$, the triplet $(u^*, y^*, \lambda^*)$ satisfies:
\begin{align}
u^*(t) &= -R^{-1}B^\top\lambda^*(t), \label{eq:pmp_opt} \\
\dot{y}^*(t) &= Ay^*(t) + Bu^*(t) + f(t), ~ y^*(0) = \mathbb{E}[x_{i,0}], \label{eq:pmp_fwd} \\
-\dot{\lambda}^*(t) &= Qy^*(t) - Qr_x(t) + \alpha\!\bigl(e_2^\top y^*(t)-r_g(t)\bigr)e_2 \notag \\
                    &\quad + A^\top\lambda^*(t), \label{eq:pmp_adjoint} \\
\lambda^*(T) &= Q_T\!\bigl(y^*(T)-r_x(T)\bigr). \label{eq:pmp_bwd}
\end{align}
We refer to \eqref{eq:pmp_opt}--\eqref{eq:pmp_bwd} as the \emph{PMP system}.

\begin{lemma}[Uniqueness of the PMP System Solution] \label{lem:pmp} 
Under Assumptions~1 and~2, the PMP system 
\eqref{eq:pmp_opt}--\eqref{eq:pmp_bwd} admits a unique solution 
$(u^*, y^*, \lambda^*)$ with $u^* \in L^2([0,T];\mathbb{R})$, 
$y^* \in C([0,T];\mathbb{R}^2)$, and $\lambda^* \in C([0,T];\mathbb{R}^2)$. 
\end{lemma}

\begin{proof}
Let $(u, y, \lambda)$ be any solution of the PMP system 
\eqref{eq:pmp_opt}--\eqref{eq:pmp_bwd}. For the cost 
\eqref{eq:aux_cost}, the optimality condition \eqref{eq:pmp_opt} is the 
first-order stationarity condition of $J$ with respect to $u$. Since $J$ 
is convex on $L^2([0,T];\mathbb{R})$ by 
Lemma~\ref{lem:aux_existence_uniqueness}, this is sufficient for $u$ to 
be a global minimizer of $J$ 
\cite[Ch.~I, Thm.~2.3]{fleming1975deterministic}. Since $J$ is strictly 
convex, the minimizer is unique, and hence $u = u^*$ 
\cite[Ch.~I, Thm.~2.4]{fleming1975deterministic}.

With $u = u^*$ fixed, the state equation \eqref{eq:pmp_fwd} is a linear 
ODE with a prescribed initial condition, and therefore admits a unique 
solution, giving $y = y^*$. Given $y^*$, the adjoint equation \eqref{eq:pmp_adjoint} with terminal 
condition \eqref{eq:pmp_bwd} is a linear ODE with continuous coefficients 
and a prescribed terminal condition, and therefore admits a unique 
solution, giving $\lambda = \lambda^*$.

Therefore, the PMP system \eqref{eq:pmp_opt}--\eqref{eq:pmp_bwd} admits 
exactly one solution triplet $(u^*, y^*, \lambda^*)$.
\end{proof}

\begin{theorem}[Existence \& Uniqueness - TPBVP] \label{thm:tpbvp_equivalence}
Under Assumptions~1 and~2, the TPBVP \eqref{eq:mfgs_adjoint}--\eqref{eq:mfgs_mean_initial} admits a unique solution pair $(\bar{x}, s)$.
\end{theorem}

\begin{proof}
Let $P(t)$ be the unique solution of \eqref{eq:mfgs_riccati_P}. Define 
the linear map $\mathcal{L}: (y,\lambda) \mapsto (\bar{x}, s)$ by
\[
\bar{x}(t) = y(t), \qquad s(t) = \lambda(t) - P(t)y(t),
\]
with inverse $\mathcal{L}^{-1}: (\bar{x}, s) \mapsto (y, \lambda)$ given by 
$y(t) = \bar{x}(t)$, $\lambda(t) = P(t)\bar{x}(t) + s(t)$.

Let $(u^*, y^*, \lambda^*)$ denote the optimal control, optimal state, 
and costate trajectory guaranteed by Lemma~\ref{lem:aux_existence_uniqueness} and Lemma~\ref{lem:pmp}. Set $s(t) := \lambda^*(t) - P(t)y^*(t)$.
Differentiating this relation and substituting the costate equation 
\eqref{eq:pmp_adjoint}, the Riccati equation \eqref{eq:mfgs_riccati_P}, and 
the state equation \eqref{eq:pmp_fwd} with $u^*(t) = -R^{-1}B^\top\lambda^*(t)$ 
yields the adjoint equation \eqref{eq:mfgs_adjoint}. The terminal condition 
$s(T) = -Q_T r_x(T)$ follows from $P(T) = Q_T$ and \eqref{eq:pmp_bwd}. Substituting 
$\lambda^* = Py^* + s$ into the state equation recovers the forward equation 
\eqref{eq:mfgs_mean_forward}. Hence $(\bar{x}^*, s^*) = \mathcal{L}(y^*, \lambda^*)$ 
satisfies the TPBVP.

Conversely, let $(\bar{x}, s)$ solve the 
TPBVP. Define $\lambda(t) := P(t)\bar{x}(t) + s(t)$ and 
$u(t) := -R^{-1}B^\top\lambda(t)$. Differentiating $\lambda$ and substituting 
the Riccati equation \eqref{eq:mfgs_riccati_P} and the TPBVP equations 
recovers the costate equation \eqref{eq:pmp_adjoint}, while the forward TPBVP 
equation \eqref{eq:mfgs_mean_forward} directly gives \eqref{eq:pmp_fwd}. 
Hence $\mathcal{L}^{-1}(\bar{x}, s)$ satisfies the PMP conditions of 
Lemma~\ref{lem:pmp}.

Since $\mathcal{L}$ is bijective with explicit inverse, it is an isomorphism. 
Lemma~\ref{lem:aux_existence_uniqueness} guarantees a unique optimal 
control $u^*$, and Lemma~\ref{lem:pmp} gives the unique triplet 
$(u^*, y^*, \lambda^*)$,
which $\mathcal{L}$ maps to a unique TPBVP pair 
$(\bar{x}^*, s^*)$, establishing both existence and uniqueness.
\end{proof}

\section{Mean Field-Affine Pricing}\label{sec:affine_pricing}

A  simplification of the MFG equilibrium is possible when the price  is specialized to an affine function.

\begin{assumption}
\label{ass:affine_price}
The price function $\alpha: \mathbb{R} \to \mathbb{R}$ takes the form
\begin{equation*}
    \alpha(d) = c_1 d + c_0, \quad \forall d \in \mathbb{R},\quad \text{with } c_1>0,~  c_0 \in \mathbb{R}.
\end{equation*}
\end{assumption}

The coefficient $c_1 > 0$ ensures that price function is  monotonically increasing as  required in Assumption \ref{ass:price_function} and $c_0 \in \mathbb{R}$ is a constant offset representing a baseline price.

\subsection{Decoupling TPBVP}
To decouple the TPBVP  \eqref{eq:mfgs_adjoint}--\eqref{eq:mfgs_mean_forward}, we 
follow the standard approach for LQ-MFGs (see e.g. \cite[p.~525]
{Bensoussan2016}) by introducing an affine transformation from the mean state to the costate 
\begin{equation}\label{decoupling_ansatz}
s(t) = \Pi(t)\bar{x}(t) + \beta(t) .
\end{equation}
Inserting \eqref{decoupling_ansatz} into equations \eqref{eq:mfgs_adjoint} and \eqref{eq:mfgs_mean_forward}, and differentiating with respect to time, we obtain $\dot{s}(t) = \dot{\Pi}(t)\bar{x}(t) + \Pi(t)\dot{\bar{x}}(t) + \dot{\beta}(t).$ Substituting this expression into equation \eqref{eq:mfgs_adjoint} and matching coefficients of $\bar{x}(t)$ and the constant terms yields two decoupled ordinary differential equations (ODEs) for $\Pi(t)$ and $\beta(t)$ as follows
\begin{equation}\label{eq:mfgs_riccati_Pi}
\begin{aligned}
-\dot{\Pi}(t)
&= \Pi(t)A_{\mathrm{cl}}(t) + A_{\mathrm{cl}}(t)^\top\Pi(t)\\[0.5ex]
&\quad
- \Pi(t)BR^{-1}B^\top\Pi(t)  + c_1 e_2 e_2^\top,
\quad
\Pi(T)=0.
\end{aligned}
\end{equation}
\begin{equation}\label{eq:mfgs_ode_beta}
\begin{aligned}
-\dot{\beta}(t)
&= \big(A_{\mathrm{cl}}(t)^\top - \Pi(t)BR^{-1}B^\top\big)\beta(t) \\[0.5ex]
&\quad
+ \big(P(t) + \Pi(t)\big)f(t) - Qr_x(t) - c_1 r_g(t)e_2 + c_0e_2 \\[0.5ex]
&
\beta(T) = -Q_T r_x(T),
\end{aligned}
\end{equation}
where $A_{\mathrm{cl}}(t)$ is the closed-loop matrix defined in \eqref{eq:closed_loop_matrix}.
Substituting $s(t)$ in
 \eqref{eq:mfgs_mean_forward} by the affine transformation in \eqref{decoupling_ansatz} yields
\begin{equation}\label{eq:mfgs_ode_xbar}
\begin{aligned}
\dot{\bar{x}}(t)
&= \big(A_{\mathrm{cl}}(t) - BR^{-1}B^\top\Pi(t)\big)\bar{x}(t) \\[0.5ex]
&\quad - BR^{-1}B^\top\beta(t) + f(t),
\qquad
\bar{x}(0)=\mathbb{E}[x_{i,0}].
\end{aligned}
\end{equation}
To further simplify the solution, we follow the idea in \cite{ZhuGao2026, GaoMalhame2025} to introduce $ \Omega(t) = P(t) + \Pi(t).$
Summing both sides of equation \eqref{eq:mfgs_riccati_P} and equation \eqref{eq:mfgs_riccati_Pi}, we obtain 
\begin{equation}\label{eq:riccati_Omega}
\begin{aligned}
-\dot{\Omega}(t)
&= A^\top\Omega(t) + \Omega(t)A - \Omega(t)BR^{-1}B^\top\Omega(t) \\
&\quad + Q + c_1 e_2 e_2^\top, \quad \Omega(T) = Q_T.
\end{aligned}
\end{equation}
\begin{remark}

The Riccati equations for both $P(t)$ and $\Omega(t)$ are standard, and admit unique and bounded solutions on $[0, T]$, since $c_1 > 0$ in Assumption~\ref{ass:affine_price} and $Q\geq0$ ensure that the effective weighting matrix $Q_{\mathrm{eff}} = Q + c_1 e_2 e_2^\top$ is positive semi-definite. Consequently, $\Pi(t) = \Omega(t) - P(t)$ is also uniquely defined for all $t \in [0,T]$, without additional restrictions on the coupling strength $c_1$ or time horizon $T$.
\end{remark}
Therefore, the MFG strategy for a generic agent $i$  is equivalently given by the feedback law 
$
u_i^\star(t) = \gamma^\star(t, x_i^\star(t)), $
with
\begin{equation}
    \gamma^\star(t,z) = -R^{-1}B^\top\big(P(t)z +  (\Omega(t)-P(t))\bar{x}(t) + \beta(t)\big),
    \label{eq:app_optimal_control-simple}
\end{equation}
where $\Omega$ is given by \eqref{eq:riccati_Omega}, $P$ is given by \eqref{eq:mfgs_riccati_P}, and
$\beta$ and $\bar{x}$ are respectively given by
\begin{equation}\label{eq:mfgs_ode_beta-simple}
\begin{aligned}
-\dot{\beta}(t)
&= \big(A^\top - \Omega(t)BR^{-1}B^\top\big)\beta(t) \\
&\quad + \Omega(t)f(t) - Qr_x(t) - c_1 r_g(t)e_2 + c_0e_2, \\
&\quad \beta(T) = -Q_T r_x(T).
\end{aligned}
\end{equation}
\begin{equation}\label{eq:mfgs_ode_xbar-simple}
\begin{aligned}
\dot{\bar{x}}(t)
&= \big(A - BR^{-1}B^\top\Omega(t)\big)\bar{x}(t) \\[0.5ex]
&\quad - BR^{-1}B^\top\beta(t) + f(t),
\qquad
\bar{x}(0)=\mathbb{E}[x_{i,0}].
\end{aligned}
\end{equation}

\begin{remark}
    The impact of the grid reference $r_g$ and the baseline price $c_0$ on the control is through $\beta$. The impact of the price sensitivity $c_1$ on the control is through $\Omega$, $\beta$, and $\bar{x}$.
\end{remark}

\begin{proposition}
\label{prop:existence-uniqueness-2}
Under
Assumptions~\ref{ass:reference_signals}
and \ref{ass:affine_price}, the MFG equilibrium strategy exists and is uniquely given by \eqref{eq:app_optimal_control-simple} together with \eqref{eq:mfgs_riccati_P}, \eqref{eq:riccati_Omega}, \eqref{eq:mfgs_ode_beta-simple}, and \eqref{eq:mfgs_ode_xbar-simple}, for any positive coupling strength $c_1 > 0$ and any time horizon $T > 0$.
\end{proposition}
\begin{proof}
    The Riccati equations for $P(t)$ and $\Omega(t)$ admit unique, bounded, and positive semi-definite solutions in $C_1([0, T]; \mathbb{R}^{2\times2})$ by standard results in linear-quadratic optimal control theory \cite[ch.~6, Cor.~2.10]{YongZhou1999}.
    With $P(t)$ and $\Omega(t)$ uniquely determined, the coefficients in the linear ODEs for $\beta(t)$ and $\bar{x}(t)$ are bounded and continuous, and hence unique and bounded solutions for $\beta(t)$ and $\bar{x}(t)$ exist on $C_1([0, T];\mathbb{R}^2)$ by standard ODE theory \cite[Corollary~2.6]{Teschl2012}.
    Hence the MFG strategy in \eqref{eq:app_optimal_control-simple} is uniquely determined. 
    This completes the proof.
\end{proof}

\begin{remark}
In general LQ-MFG problems, the existence and uniqueness of an equilibrium are typically guaranteed only under additional assumptions, such as a sufficiently small coupling parameter or a sufficiently short time horizon (see e.g. \cite[Thm. 3.3]{Bensoussan2016}, \cite{HuangCainesMalhame2007}).
In contrast, in our current work, no additional assumptions on the Riccati equation are needed for the existence and uniqueness of the MFG equilibrium. This difference arises because the cost for an individual agent in \eqref{eq:limit-cost} does not contain the penalty of quadratic mean field tracking error as in the LQ-MFG literature \cite{Bensoussan2016,HuangCainesMalhame2007}. 
\end{remark}

\section{Simulation Results}
\label{sec:simulation}

We present a numerical study of an overnight charging scenario for $N=200$ electric
vehicles to validate the theoretical framework. Two price functions $\alpha(d)$,
where $d(t)=\bar{x}_2(t)-r_g$, are considered in this study and illustrated in
Fig.~\ref{fig:price_functions}:
\begin{equation}
\alpha_{\mathrm{sig}}(d) = \frac{d_{\max}}{1+e^{-ad}},
\qquad
\alpha_{\mathrm{aff}}(d) = c_1 d + c_0.
\label{eq:price_functions}
\end{equation}
Both functions satisfy Assumption~2. Each price structure is examined under two cost
scenarios: $Q=0$ and $Q \neq 0$.

\begin{figure}[t]
\centering
\includegraphics[width=0.98\linewidth]{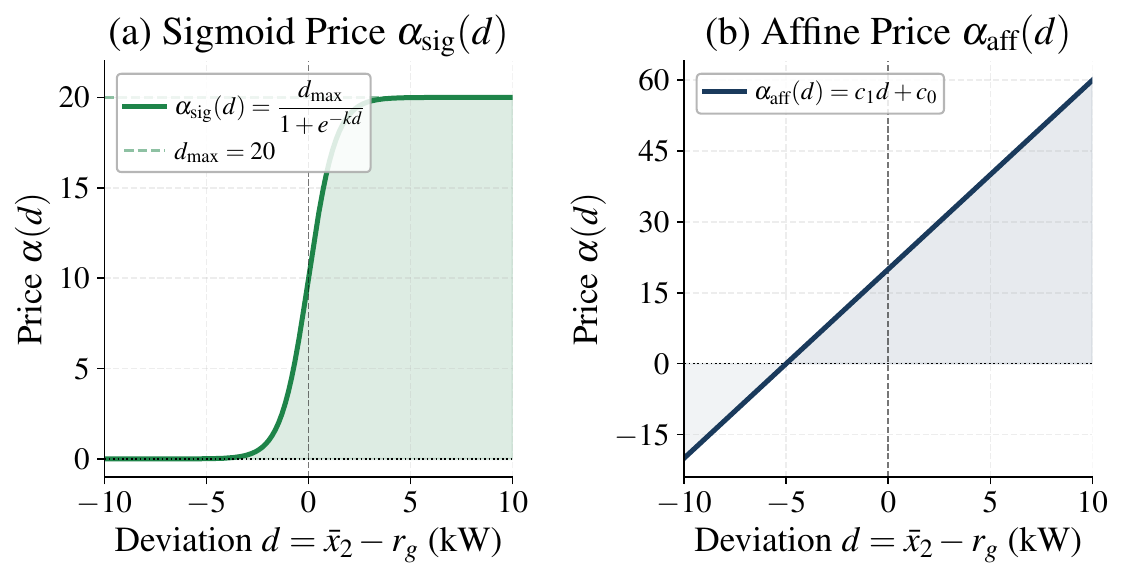}
\caption{Price functions as a function of the deviation
$d=\bar{x}_2-r_g$: (a)~sigmoid price $\alpha_{\mathrm{sig}}$;
(b)~affine price $\alpha_{\mathrm{aff}}$ with $c_0=20$.}
\vspace{-0.65 cm}
\label{fig:price_functions}
\end{figure}

\subsection{Simulation Setup}

Each vehicle has a battery capacity of $60\,\mathrm{kWh}$ and begins with a random
initial SOC drawn uniformly from $[18, 30]\,\mathrm{kWh}$, with zero initial charging
power. The individual cost targets are $r_{x,1}=54\,\mathrm{kWh}$ (90\% SOC) and
$r_{x,2}=9.6\,\mathrm{kW}$, with a terminal power reference $r_{x,2}(T)=0$. The grid
target is $r_g = 5\,\mathrm{kW}$; \textit{this is selected as an approximate value consistent with the
energy required for the mean agent to reach its terminal SOC over the horizon, based on
the average initial SOC.} The stochastic dynamics are integrated
via Euler--Maruyama with $\Delta t=0.005\,\mathrm{h}$. Key parameters are listed in
Table~\ref{tab:parameters}.

\begin{table}[h!]
\centering
\caption{Simulation Parameters}
\label{tab:parameters}
\begin{tabular}{llll}
\hline
\textbf{Parameter} & \textbf{Value} & \textbf{Parameter} & \textbf{Value} \\
\hline
$N$               & 200          & $\sigma_1$          & 0.5 \\
$T$               & 8 h          & $\sigma_2$          & 0.25 \\
$E_{\mathrm{cap}}$& 60 kWh       & $Q$                 & $\mathrm{diag}(0.5,\;2.5)$ \\
$\kappa$          & 0.9          & $R$                 & 0.10 \\
$r_{x,1}$         & 54 kWh       & $Q_T$               & $\mathrm{diag}(60,\;1)$ \\
$r_{x,2}$         & 9.6 kW       & $r_{x,2}(T)$        & 0 kW \\
$r_g$             & 5.0 kW       & $\Delta t$          & 0.005 h \\
$c_1$             & 4            & $a$                 & 1.5 \\
$c_0$             & 20           & $d_{\max}$          & 20 \\
\hline
\end{tabular}
\end{table}
\begin{figure}[htb!]
\centering
\includegraphics[width=0.98\linewidth]{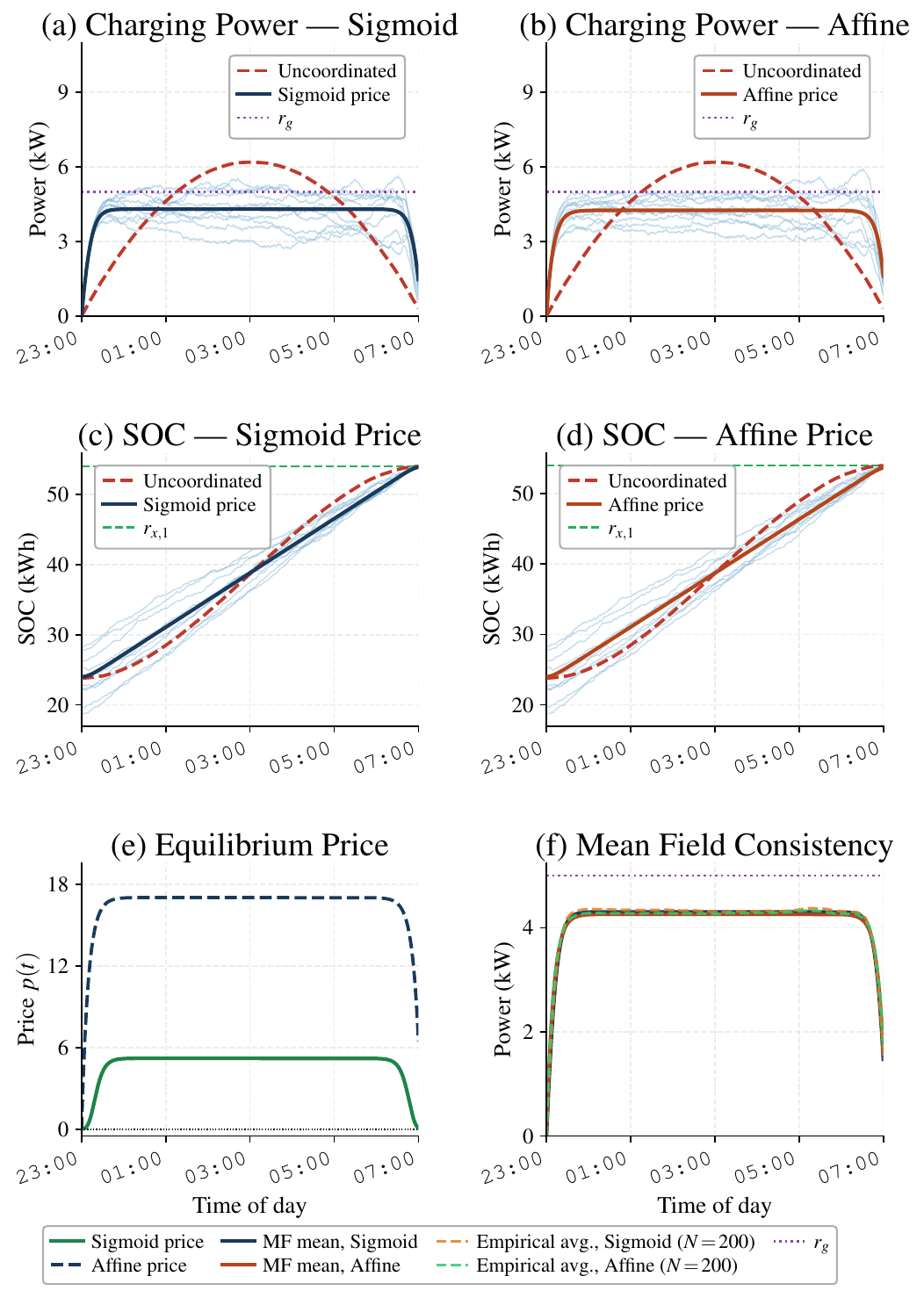}
\caption{MFG equilibrium results for $Q = 0$.
(a)--(b) Optimal charging power under sigmoid and affine price coordination (solid),
and uncoordinated profile (red dashed).
(c)--(d) Corresponding SOC trajectories.
(e) Equilibrium price signals. (f) Mean field consistency verification.}
\vspace{-0.4 cm}
\label{fig:q0_results}
\end{figure}
\subsection{Results: Pure Price Coordination ($Q=0$)}

We first consider the case $Q=0$, in which the quadratic 
state-tracking term is absent and the equilibrium is governed 
by the price signal and the control effort penalty. The results 
are shown in Fig.~\ref{fig:q0_results}.
In both the affine and sigmoid cases, the mean charging power converges to a near-constant plateau, as the price signal drives $\bar{x}_2(t)$ toward $r_g$ with no state-tracking penalty. In contrast, the uncoordinated strategy produces a bell-shaped 
power profile that peaks well above $r_g$, illustrating the 
benefit of price coordination for grid compliance.
The corresponding SOC trajectories in panels~(c)--(d) are approximately linear,
consistent with the near-constant charging rate, and all agents converge to the terminal
target $r_{x,1}$. Regarding the equilibrium price in panel~(e), the sigmoid price
maintains a low, nearly constant level close to $r_g$, while the affine price holds a
significantly higher plateau, reflecting the larger linear penalty imposed by
$\alpha_{\mathrm{aff}}$ when the mean field deviation is non-negligible. The empirical
mean $\bar{x}_2^N(t)$ closely tracks the theoretical mean field $\bar{x}_2(t)$ in
panel~(f), validating the mean field approximation.

\begin{figure}[htb!]
\centering
\includegraphics[width=0.98\linewidth]{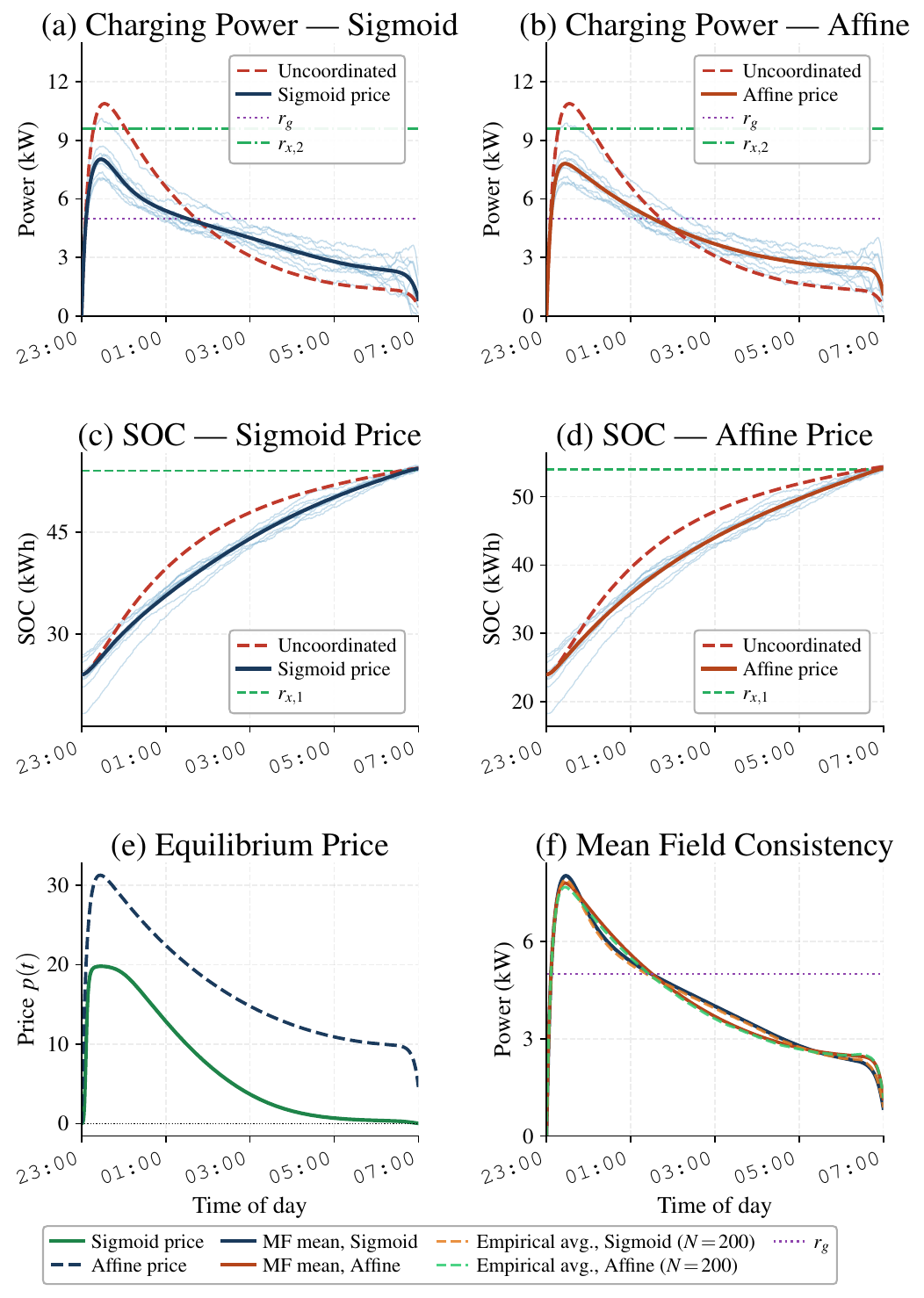}
\caption{MFG equilibrium results for $Q \neq 0$.
(a)--(b) Optimal charging power under sigmoid and affine price coordination (solid),
with individual trajectories (light blue) and uncoordinated baseline (red dashed).
(c)--(d) Corresponding SOC trajectories.
(e) Equilibrium price signals. (f) Mean field consistency verification.}
\vspace{-0.5 cm}
\label{fig:q_nonzero_results}
\end{figure}
\subsection{Results: Price and Running Cost Coordination ($Q \neq 0$)}

We next include a running cost ($Q=\mathrm{diag}(0.5, 2.5)$) that penalizes deviations
from the reference trajectory $(r_{x,1}, r_{x,2})$. As shown in
Fig.~\ref{fig:q_nonzero_results}, this additional cost fundamentally alters the shape
of the equilibrium trajectory from flat to curved: the mean charging power rises
sharply at the start of the session, peaks near the beginning of the horizon, and then
decays gradually as agents approach their terminal targets. The uncoordinated baseline
again exhibits a higher and broader peak, underscoring the role of the price signal in
redistributing the charging load. SOC trajectories in panels~(c)--(d) are concave and converge to $r_{x,1}$, consistent with the decreasing charging rate. Individual trajectory spread in panels~(a)--(d) reflects the stochastic noise in the dynamics ($\sigma_1 = 0.5$, $\sigma_2 = 0.25$). A key structural difference between the two
price functions emerges in the terminal behavior shown in panel~(e). The affine price,
being unbounded, linearly amplifies the terminal transient in $\bar{x}_2(t)$ required
to meet the boundary condition, producing a sharp price spike near $t=T$. In contrast,
the sigmoid price saturates at $d_{\max}$, which prevents such volatility and ensures a
stable terminal price. The mean field consistency in panel~(f) confirms that
$\bar{x}_2^N(t)$ closely tracks $\bar{x}_2(t)$ under both price functions.

Both price structures yield a valid MFG equilibrium in both cost scenarios, confirming the framework's generality. The sigmoid price is suited when a bounded, non-negative signal is required; the affine price applies when an unrestricted linear response is admissible.

\section{CONCLUSION}

This paper developed a price-coordinated MFG framework for the decentralized charging of large-scale battery populations. One key feature of the model is the treatment of charging power as a state variable.
The existence and uniqueness of the MFG equilibrium were established for any continuous and monotonically increasing price function, a result that holds for any finite time horizon without additional restrictions on the coupling strength. For the special case of an affine price function, a simplified representation of the MFG equilibrium was derived based on two decoupled Riccati equations. Future work will focus on proving the $\varepsilon$-Nash property of the MFG strategy and price-coordinated MFG problems and on extending the framework to accommodate hard constraints on states and control actions.

\section{Acknowledgments} 
The first author acknowledges that ChatGPT was used  to improve the writing of the manuscript, identify the function $\Phi$ in \eqref{eq:aux_cost}, and assist with the coding. All technical content, results, and interpretations were reviewed and validated by the authors.

\bibliographystyle{IEEEtran}
\bibliography{ifacconf}

\end{document}